\documentclass[twoside]{aiml20}
\usepackage{ifthen}
\newboolean{full}
\setboolean{full}{true} 

\usepackage{aiml20macro}
\newenvironment{tremark}[1][]{\vspace{-\lastskip}\par \addvspace{\thmvspace}\begin{rem}[#1] \rm
  }{\end{rem}\par\addvspace{\thmvspace}}

\usepackage{graphicx}
\usepackage{amsmath}
\usepackage{amssymb}
\usepackage{enumitem}

\usepackage{stmaryrd}

\usepackage{calrsfs}

\usepackage{hyperref}

\usepackage{apptools}

\usepackage[utf8]{inputenc}

\newtheorem{lemdefn}[thm]{Lemma and Definition}

\newenvironment{lemmadefn}{\vspace{-\lastskip}\par\addvspace{\thmvspace}\begin{lemdefn}}{\end{lemdefn}\par\addvspace{\thmvspace}}

\newcommand{\Land}{\bigwedge}
\newcommand{\zeroone}{{0\hspace{1pt}\text{-}1}}

\newcommand{\into}{\hookrightarrow}
\newcommand{\by}[1]{\text{(#1)}}

\newcommand{\Sem}[1]{\llbracket #1\rrbracket}
\newcommand{\Rat}{\mathbb{Q}}
\newcommand{\Set}{\mathsf{Set}}
\newcommand{\Op}{\mathsf{op}}
\newcommand{\id}{\mathsf{id}}
\newcommand{\Struct}{\mathcal{M}}
\newcommand{\Lang}{\mathcal{L}}
\newcommand{\form}{\mathcal{F}}
\newcommand{\Axioms}{\mathcal{A}}
\newcommand{\Prop}{\mathsf{Prop}}
\newcommand{\nondash}[1][\Lang]{\vdash_{#1}^{\zeroone}}
\newcommand{\ldash}{\vdash_\Lang}
\newcommand{\pow}{\mathcal{P}}
\newcommand{\contrapow}{\mathcal{Q}}
\newcommand{\neighb}{\mathcal{N}}
\newcommand{\ultra}{\mathcal{U}}
\newcommand{\Dist}{\mathcal{D}}

\DeclareMathAlphabet{\pazocal}{OMS}{zplm}{m}{n}
\DeclareMathAlphabet{\mathcal}{OMS}{cmsy}{m}{n}

\def\lastname{Forster and Schröder}

\begin{document}

\begin{frontmatter}
  \title{Non-iterative Modal Logics are Coalgebraic} \author{Jonas Forster \and Lutz Schröder}
  \address{Friedrich-Alexander-Universität Erlangen-Nürnberg}
  
  \begin{abstract} A modal logic is \emph{non-iterative} if it can be
    defined by axioms that do not nest modal operators, and
    \emph{rank-1} if additionally all propositional variables in
    axioms are in scope of a modal operator. It is known that every
    syntactically defined rank-1 modal logic can be equipped with a
    canonical coalgebraic semantics, ensuring soundness and strong
    completeness. In the present work, we extend this result to
    non-iterative modal logics, showing that every non-iterative modal
    logic can be equipped with a canonical coalgebraic semantics
    defined in terms of a copointed functor, again ensuring soundness
    and strong completeness via a canonical model construction. Like
    in the rank-1 case, the canonical coalgebraic semantics is
    equivalent to a neighbourhood semantics with suitable frame
    conditions, so the known strong completeness of non-iterative
    modal logics over neighbourhood semantics is implied. As an
    illustration of these results, we discuss deontic logics with
    factual detachment, which is captured by axioms that are
    non-iterative but not rank~1.
  \end{abstract}
  
\end{frontmatter}

\section{Introduction}

Modal frame axioms are called \emph{non-iterative} if they do not nest
modal operators, and \emph{rank-1} if additionally all occurrences of
propositional variables are under modal operators; logics are
non-iterative or rank-1, respectively, if they can be axiomatized by
axioms of the correspondingly restricted shape. Prominent examples
include the $\text{K}$-axiom $\Box (a\to b)\to \Box a\to\Box b$, which
is rank-1, and the $\text{T}$-axiom $\Box a \rightarrow a$, which is
non-iterative. Previous work in coalgebraic
logic~\cite{DBLP:journals/logcom/SchroderP10} shows that every rank-1
modal logic is strongly complete over a canonical coalgebraic
semantics that can be seen to coincide with neighbourhood
semantics. In the present paper, we extend this result to
non-iterative logics: We show that every (syntactically given)
non-iterative modal logic is strongly complete over a canonical
coalgebraic semantics, which again turns out to coincide with
neighbourhood semantics, so that the known result that non-iterative
logics are complete over their neighbourhood
semantics~\cite{Surendonk97} is implied.

Generally, the semantic framework of \emph{coalgebraic
  logic}~\cite{DBLP:journals/cj/CirsteaKPSV11} supports general
proof-theoretic, algorithmic, and meta-theoretic results that can be
instantiated to the logic of interest, cutting out much of the
repetitive labour associated with the iterative process of designing
an application-specific logic. The framework is based on casting
state-based models of various types (e.g.\ relational, probabilistic,
neighbourhood-based, or game-based) as coalgebras for a functor, the
latter to be thought of as encapsulating the structure of the
successors of a state.

It has been shown that the modal logic of the class of \emph{all}
coalgebras for a given functor can always be axiomatized in
rank~1~\cite{DBLP:journals/jlp/Schroder07}. Conversely, as indicated
above, every rank-1 logic has a coalgebraic
semantics~\cite{DBLP:journals/logcom/SchroderP10}; roughly speaking,
rank-1 axioms can be absorbed into a functor. The coalgebraic
treatment of \emph{non-iterative} axioms thus requires a
generalization to copointed functors, to be thought of as
incorporating the present state as well as its successors. Indeed it
turns out that to obtain strong completeness, it is useful to
generalize further to \emph{weakly} copointed functors, in which the
present state is virtualized as an ultrafilter, and subsequently
restrict to \emph{proper} coalgebras of such weakly copointed
functors, in which all these virtual points actually materialize. Our
main result thus states more precisely that every non-iterative logic
is sound and strongly complete for the class of proper coalgebras of a
canonical weakly copointed functor we construct; strong completeness
w.r.t.\ a canonical copointed subfunctor then follows. As indicated
above, this result translates back to imply strong completeness
w.r.t.\ neighbourhood semantics as originally proved by
Surendonk~\cite{Surendonk97}. Our proof differs quite markedly from
Surendonk's; while the latter makes central use of first-order model
theory (specifically, compactness), we avoid compactness and instead
work with solution theorems in Boolean algebra. We complement strong
completeness of the canonical coalgebraic semantics with an (easier)
result showing that the modal logic of a copointed functor can always
be equipped with a weakly complete non-iterative axiomatization,
justifying the slogan that non-iterative logics are precisely the
logics of copointed functors.

We illustrate the use of this result on certain deontic logics that on
the one hand avoid the deontic explosion problem (ruling out
normality, and hence Kripke semantics) and on the other hand allow for
\emph{factual detachment}, embodied in properly non-iterative
axioms~\cite{DBLP:journals/japll/Strasser11}. The only known semantics
for such logics is neighbourhood semantics. Weak completeness and the
finite model property follow from the previous results by
Lewis~\cite{DBLP:journals/jphil/Lewis74}, alternatively by a concrete
proof given in the online appendix
of~\cite{DBLP:journals/japll/Strasser11}. Moreover, the cited result
by Surendonk~\cite{Surendonk97} implies strong completeness. Our
present results reprove strong completeness, and provide a
coalgebraization of the semantics in terms of a copointed functor.

This paper is a full version of a conference
abstract~\cite{ForsterSchroder20}; we note that the conference
abstract misses reference~\cite{Surendonk97}.

\paragraph{Organization} We recall the syntactic notion of
non-iterative modal logic~\cite{DBLP:journals/jphil/Lewis74} in
Section~\ref{sec:non-iterative}. In Section~\ref{sec:semantics}, we
recall the semantic framework of coalgebraic logic, and discuss
copointed and weakly copointed functors. Our main technical tool is
the 0-1-step logic of a non-iterative coalgebraic logic, introduced in
Section~\ref{sec:zeroone-logic}. We establish the easier direction of
the relationship between non-iterative modal logics and coalgebraic
modal logic in Section~\ref{sec:fmp}, where we show that the modal
logic of coalgebras for a copointed functor is always non-iterative
(and has the finite model property). Our main result, which states
that conversely, every non-iterative modal logic is strongly complete
over a canonical coalgebraic semantics, is shown in
Sections~\ref{sec:can-struct} and~\ref{sec:str-compl}. In
Section~\ref{sec:deontics}, we present applications to deontic logics.
\iffull Some proofs are deferred to Appendix~\ref{app:proofs}. \else
Proofs that are omitted or only sketched can be found in the full
version~\cite{ForsterSchroder20arXiv}.\fi

\section{Non-iterative Modal Logics}\label{sec:non-iterative}

A \textit{(modal) similarity type} $\Lambda$ is a set of modal
operators with associated finite arity. The set $\form(\Lambda)$ of
\emph{$\Lambda$-formulae} is given by the grammar
\begin{equation*}
\phi_1 , \dots , \phi_n ::= \bot \mid \neg \phi_1 \mid \phi_1 \wedge
\phi_2 \mid L(\phi_1, \dots , \phi_n)
\end{equation*}
where $ L \in \Lambda $ has arity $n$.  Additional Boolean operators
$\rightarrow$, $\leftrightarrow$, $ \vee $ and $ \top $ can then be
defined as usual. We denote by~$|\phi|$ the \emph{size} of a
formula~$\phi$, measured as the number of subformulae of~$\phi$. The
grammar does not include propositional atoms as a separate syntactic
category; however, these can be cast as nullary modalities. We thus
distinguish propositional atoms from propositional variables, which
are used to formulate axioms and rules.

\begin{definition} Let $\Prop(V)$ denote the set of propositional
  formulae~$\phi$ over a given set~$V$ (i.e.\
  $\phi::=\bot\mid a\mid\neg\phi\mid\phi_1\land\phi_2$, with
  $a$ ranging over~$V$), and put
  \begin{equation*}
    \Lambda(V)=\{ L(a_1, \ldots, a_n) \mid L \in \Lambda \text{ n-ary}, a_1,
    \ldots, a_n \in V\}.
  \end{equation*}
  The elements of~$V$ are typically thought of as propositional
  variables.  A \textit{one-step formula} or \emph{rank-1 formula}
  over~$V$ is a formula in $\Prop(\Lambda(\Prop(V)))$, and a
  \textit{0-1-step formula} or a \textit{non-iterative formula}
  over~$V$ is a formula in $\phi \in \Prop(\Lambda(\Prop(V)) \cup
  V)$. In words, a formula $\phi$ over~$V$ is non-iterative if it does
  not contain nested modal operators, and a non-iterative
  formula~$\phi$ is rank-1 if additionally every variable in~$\phi$
  lies under a modal operator. We generally refer to maps of the form
  $\sigma\colon V\to Z$ that we use to replace entities of type~$V$
  with entities of type~$Z$ in formulae as \emph{$Z$-substitutions
    on~$V$}, and write $\phi\sigma$ for the result of
  applying~$\sigma$ to a formula~$\phi$ over~$V$.
\end{definition}
\noindent As indicated previously, the axiom
$\Box(a\to b)\to\Box a\to\Box b$ (with $a,b$ propositional variables)
is rank-1, and $\Box a\to a$ is non-iterative. We define modal logics
$\Lang = (\Lambda, \Axioms)$ syntactically by a similarity
type~$\Lambda$ and a set~$\Axioms$ of \emph{axioms} (in the given
similarity type), determining the set of derivable formulae via the
usual proof system as recalled below.  A logic is \emph{non-iterative}
(\emph{rank-1}) if all its axioms are non-iterative (rank-1). Given a
logic $\Lang = (\Lambda, \Axioms)$, we say that a
$\Lambda$-formula~$\psi$ is \emph{derivable}, and write $\ldash \psi$,
if~$\psi$ can be derived in finitely many steps via the following
rules:

\begin{gather*}(Ax)\;\frac{}{\psi\sigma}(\psi \in \Axioms,
  \sigma \text{ an }\form(\Lambda)\text{-substitution})\\
  (P)\;\frac{\phi_1\quad \ldots\quad \phi_n}{\psi}(\{\phi_1, \ldots, \phi_n\}
  \vdash_{PL} \psi)\quad\; (C)\;\frac{\phi_1 \leftrightarrow \psi_1\quad
    \ldots\quad \phi_n \leftrightarrow \psi_n}{L(\phi_1, \ldots, \phi_n)
    \leftrightarrow L(\psi_1, \ldots, \psi_n)}
\end{gather*}
where by $\{\phi_1, \ldots, \phi_n\} \vdash_{PL} \psi$ we indicate
that~$\psi$ is derivable from assumptions $\phi_1, \ldots, \phi_n$ by
propositional reasoning (e.g.\ propositional tautologies and modus
ponens). The last rule is known as the \emph{congruence rule} or
\emph{replacement of equivalents}. For a set~$\Phi$ of
$\Lambda$-formulae, we write
$\Phi \vdash_{\Lang} \psi$ if
$\ldash (\phi_1 \wedge \ldots \wedge \phi_n)\rightarrow \psi$ for some
$\phi_1, \ldots, \phi_n \in \Phi$. We say that~$\Phi$ is
\emph{$\Lang$-consistent}, or just \emph{consistent}, if
$\Phi \not \ldash \bot$. A formula~$\phi$ is \emph{consistent} if
$\{\phi\}$ is consistent.

\begin{remark}\label{rem:rules}
  Non-iterative logics can alternatively be presented in terms of
  proof rules: A \emph{non-iterative rule} $\phi/\psi$ over~$V$
  consists of a \emph{premiss} $\phi\in\Prop(V)$ and a
  \emph{conclusion} $\psi\in\Prop(\Lambda(V)\cup V)$. There are mutual
  conversions between the two formats, the conversion from axioms to
  rules being straightforward, and the conversion from rules to axioms
  being based on Boolean unification: Given a non-iterative rule
  $\phi/\psi$, pick a \emph{projective
    unifier}~\cite{DBLP:journals/logcom/Ghilardi97} of~$\phi$, i.e.\ a
  substitution~$\sigma$ such that $\phi\sigma$ and
  $\phi\to (a\leftrightarrow\sigma(a))$, for all variables~$a$
  in~$\psi$, are tautologies, and replace $\phi/\psi$ with the axiom
  $\psi\sigma$; further details are as in the rank-1
  case~\cite{DBLP:journals/jlp/Schroder07}.
\end{remark}

\begin{remark}
  As the syntax of the logic itself does not include propositional
  variables, the above system also does not derive formulae with
  variables. If desired, propositional variables in formulae can be
  emulated by introducing fresh propositional atoms (treated as
  nullary modal operators as indicated above). In particular, if
  substitution is made to apply to these fresh propositional atoms, then
  the standard substitution rule $\phi/\phi\sigma$ becomes admissible.
\end{remark}

\section{Coalgebraic Semantics}\label{sec:semantics}

We next recall basic definitions in universal
coalgebra~\cite{DBLP:journals/tcs/Rutten00} and coalgebraic
logic~\cite{DBLP:journals/cj/CirsteaKPSV11}, which will form the
underlying semantic framework for our main result. We briefly recall
requisite categorical definitions; some familiarity with basic
category theory will nevertheless be helpful (e.g.~\cite{Awodey10}).

The underlying principle of (set-based) universal coalgebra is to
encapsulate a type of state-based systems as an endofunctor
$T: \Set \rightarrow \Set$ (briefly called a \emph{set functor}) where
$\Set$ is the category of sets and functions. Thus,~$T$ assigns to
each set~$X$ a set~$TX$, and to each map $f\colon X\to Y$ a map
$Tf\colon TX\to TY$, preserving identities and composition. We think
of $TX$ as a type of structured collections over~$X$. A basic example
is the \emph{(covariant) powerset functor}~$\pow$, which assigns to
each set~$X$ its powerset $\pow X$, and to each map~$f\colon X\to Y$
the map $\pow f\colon\pow X\to\pow Y$ that takes direct images, i.e.\
$\pow f(A)=f[A]$ for $A\in\pow X$. The most relevant example for our
present purposes is the \emph{neighbourhood functor}~$\neighb$,
defined as follows. The \emph{contravariant powerset
  functor}~$\contrapow$ is a functor of type $\Set^\Op\to\Set$, i.e.\
reverses the direction of maps; it maps a set~$X$ to its powerset
$\contrapow X=\pow X$, and a map $f\colon X\to Y$ to the preimage map
$\contrapow f\colon\contrapow Y\to\contrapow X$, i.e.\
$\contrapow f(B)=f^{-1}[B]$ for $B\in\contrapow Y$. For any
functor~$F$, we indicate by $F^\Op$ the functor that acts like~$F$ but
on the opposite categories, i.e.\ with arrows reversed in both domain
and codomain. Then, we define $\neighb$ as the composite
\begin{equation*}
  \neighb=\contrapow\circ\contrapow^\Op\colon\Set\to\Set.
\end{equation*}
We think of elements of~$\neighb X$ as neighbourhood systems over~$X$.

Given a functor~$T$, systems are then abstracted as
\emph{$T$-coalgebras} $C = (X, \xi)$ consisting of a set~$X$ of
\emph{states} and a \emph{transition function}
$\xi : X \rightarrow TX$. We think of~$\xi$ as assigning to each
state~$x$ a structured collection~$\xi(x)$ of successors. E.g.\
$\pow$-coalgebras are just Kripke frames, assigning as they do to each
state a set of successors, and $\neighb$-coalgebras are neighbourhood
frames, where each state receives a collection of neighbourhoods.

Modal operators are semantically interpreted by \emph{predicate
  liftings}~\cite{DBLP:journals/ndjfl/Pattinson04,DBLP:journals/tcs/Schroder08}:

\begin{definition} An $n$-ary \emph{predicate lifting} for a set
  functor $T$ is a natural transformation
  $\lambda : \contrapow^n \rightarrow \contrapow \circ T^{op}$,
  with~$\contrapow$ being the contravariant powerset functor recalled
  above. So $\lambda$ is a family of functions~$\lambda_X$, indexed
  over all sets~$X$, such that for all $f: X \rightarrow Y$
  and $B_i \subseteq Y$, $i = 1, \ldots, n$,
  $$ \lambda_X(f^{-1}[B_1], \ldots , f^{-1}[B_n]) =
  (Tf)^{-1}[\lambda_Y(B_1, \ldots, B_n)]. $$ A
  \emph{$\Lambda$-structure}
  $\Struct=(T, \llbracket L \rrbracket_{L \in \Lambda})$ for a
  signature $\Lambda$ consists of a functor $T$ and an $n$-ary
  predicate lifting $\llbracket L \rrbracket$ for every $n$-ary modal
  operator $L \in \Lambda$; we say that~$\Struct$ is \emph{based
    on~$T$}. When there is no danger of confusion, we will
  occasionally refer to the entire $\Lambda$-structure just as~$T$.
\end{definition}

\noindent Given a $\Lambda$-structure~$\Struct$ based on $T$, we define the
satisfaction relation $x\models_C\phi$ between states~$x$ in 
$T$-coalgebras $C = (X, \xi)$ and $\Lambda$-formulae~$\phi$ inductively
by
\begin{alignat*}{2}
  &x \not \models_C \bot &\\
  & x \models_C \neg\phi &&\text{ iff } x \not \models_C
                            \phi\\
  & x \models_C \phi \wedge \psi &&\text{ iff } x
                                    \models_C \phi
                                    \text{ and
                                    } x \models_C \psi\\
  & x \models_C L(\phi_1, \ldots \phi_n) &&\text{ iff } \xi(x) \in
                                            \llbracket L
                                            \rrbracket(\llbracket \phi_1
                                            \rrbracket_C, \ldots ,
                                            \llbracket \phi_n
                                            \rrbracket_C)
\end{alignat*}
where we write $\llbracket \phi \rrbracket_C$ (or just $\Sem{\phi}$)
for the \emph{extension} $ \{ x \in X \mid x \models_C \phi \}$
of~$\phi$.
\begin{example}\label{expl:logics}
  \begin{enumerate}[wide]
  \item\label{item:Kripke} As indicated above, Kripke frames are
    coalgebras for the powerset functor~$\pow$. The standard $\Box$
    modality is interpreted over~$\pow$ via the predicate lifting
  \begin{equation*}
    \llbracket \Box \rrbracket _X (A) = \{ B \in \pow X \mid B
    \subseteq A\},
  \end{equation*}
  which in combination with the above definition of the satisfaction
  relation induces precisely the usual semantics
  of~$\Box$.
\item\label{item:prob} \emph{Probabilistic modal
    logic}~\cite{DBLP:journals/iandc/LarsenS91,DBLP:journals/geb/HeifetzM01}
  has unary modal operators $L_p$ indexed over $p \in [0, 1]\cap\Rat$,
  with $L_p\phi$ read `$\phi$ holds with probability at least~$p$
  after the next transition step'. It is interpreted over
  probabilistic transition systems (or Markov chains), which are
  coalgebras for the \emph{discrete distribution functor}~$\Dist$,
  given on sets~$X$ by taking~$\Dist X$ to be the set of discrete
  probability distributions on~$X$. The modal operators are then
  interpreted using the predicate liftings
$$\llbracket L_p \rrbracket_X(A) = \{ \mu \in \Dist X \mid \mu(A) \geq p \}.$$
\item \label{nft} As seen above, \emph{neighbourhood frames} are
  coalgebras for the neighbourhood functor~$\neighb$. We capture the
  usual neighbourhood semantics of the~$\Box$ modality by the
  predicate lifting
  \begin{equation*}
    \Sem{\Box}_X(A)=\{N\in\neighb X\mid A\in N\},
  \end{equation*}
  that is, a state satisfies $\Box\phi$ iff the extension~$\Sem{\phi}$
  is a neighbourhood of~$x$.  More generally, a
  \emph{$\Lambda$-neighbourhood frame} for a similarity type~$\Lambda$
  is a pair $(X, (\nu_L)_{L\in \Lambda})$ consisting of a set $X$ of
  states and a family of functions
  $\nu_L \colon X \to \pow((\pow X)^n)$ for $L\in \Lambda$ $n$-ary. We
  refer to subsets of $(\pow X)^n$ as \emph{$n$-ary neighbourhood
    systems}, and to their elements as \emph{$n$-ary neighbourhoods};
  if $(A_1,\dots,A_n)\in\nu_L(x)$ for $n$-ary $L\in\Lambda$, then
  $(A_1,\dots,A_n)$ is an \emph{($n$-ary) $L$-neighbourhood
    of~$x$}. Satisfaction of modalized formulae by states $x\in X$ is
  then defined by
  \begin{equation*}
    x \models L(\phi_1, \ldots, \phi_n) \text{ iff } (\llbracket
    \phi_1 \rrbracket, \ldots, \llbracket \phi_n \rrbracket) \in
    \nu_L(x);
  \end{equation*}
  in words, $x\models L(\phi_1, \ldots, \phi_n)$ iff
  $(\llbracket \phi_1 \rrbracket, \ldots, \llbracket \phi_n
  \rrbracket)$ is an $L$-neighbourhood of~$x$.
  $\Lambda$-neighbourhood frames are coalgebras for the functor
  $\neighb_\Lambda$ defined by
  \begin{equation*}
    \neighb_\Lambda=\textstyle\prod_{L\in \Lambda\text{ $n$-ary}} \contrapow \circ ((\contrapow^\Op)^n)
\end{equation*}
  where product and $n$-th power $(-)^n$ are pointwise, i.e.\
  $\neighb_\Lambda X=\textstyle\prod_{L\in \Lambda\text{ $n$-ary}} \contrapow ((\contrapow X)^n)$. The
  corresponding predicate liftings are
  \begin{equation*}
  \llbracket L \rrbracket_X(A_1, \ldots, A_n) = \{(N_L)_{L\in\Lambda} \in\neighb_\Lambda X\mid (A_1, \ldots, A_n) \in N_L\}.
\end{equation*}
\end{enumerate}
\end{example}

\bigskip

\noindent Since we work with classical negation, we can reduce all
reasoning problems to satisfiability in the usual manner. Given a
$\Lambda$-structure based on~$T$, a formula~$\phi$ is \emph{valid} if
$x\models_C\phi$ for all states~$x$ in $T$-coalgebras $C$, and a
set~$\Phi$ of formulae is \emph{satisfiable} if there exists a
state~$x$ in a $T$-coalgebra $C$ such that $x\models_C\phi$ for all
$\phi\in\Phi$. A formula~$\phi$ is satisfiable if $\{\phi\}$ is
satisfiable. A logic~$\Lang=(\Lambda,\Axioms)$, or just~$\Axioms$, is
\emph{sound} for~$\Struct$ if all $L$-derivable formulae are valid
over~$\Struct$, \emph{weakly complete} if all consistent formulae are
satisfiable (equivalently all valid formulae are derivable), and
\emph{strongly complete} if all consistent sets of formulae are
satisfiable (which is equivalent to completeness w.r.t.\ local
consequence from possibly infinite sets of assumptions.)

It has been shown that coalgebraic modal logics coincide with rank-1
logics. More precisely, for every $\Lambda$-structure~$\Struct$ there
exists a rank-1 logic that is weakly complete
for~$\Struct$~\cite{DBLP:journals/jlp/Schroder07} (strong completeness
cannot be expected as coalgebraic modal logics often fail to be
compact, e.g.\ probabilistic modal logic as described in
Example~\ref{expl:logics}.\ref{item:prob} is not
compact~\cite{DBLP:journals/jlp/Schroder07}). Conversely, given a
rank-1 logic $\Lang=(\Lambda,\Axioms)$, there is a
$\Lambda$-structure~$\Struct$ such that~$\Lang$ is sound and strongly
complete for~$\Struct$~\cite{DBLP:journals/logcom/SchroderP10}; this
$\Lambda$-structure is isomorphic to neighbourhood
semantics. Roughly speaking, rank-1 axioms can be absorbed into the
functor; as a very simple example, the seriality axiom for Kripke
frames, $\neg\Box\bot$, can be captured by replacing the powerset
functor~$\pow$ with the non-empty powerset functor $\pow^\star$, where
$\pow^\star X=\{A\in\pow X\mid A\neq\emptyset\}$.

To cover non-iterative logics, we therefore need additional structure
on the functor that additionally caters for base points: A
\emph{copointed functor} $(T,\varepsilon)$, or just~$T$ when
$\varepsilon$ is clear from the context, consists of a functor~$T$ and
a \emph{copoint}~$\varepsilon$, i.e.\ a natural transformation
$\varepsilon\colon T\to\id$ where $\id$ denotes the identity
functor. Coalgebras $C=(X,\xi)$ for a copointed functor are by default
required to be \emph{proper}, i.e.\
$\varepsilon_X\circ\xi=id_X$. Intuitively, a plain functor encapsulates
only the possible (structured collections of) successors that can be
assigned to a given present state, while a copointed functor
additionally retains the information about the present state itself,
accessed via the copoint; the properness condition
$\epsilon_X\circ\xi$ on coalgebras $(X,\xi)$ of a copointed functor
effectively demands that this information is accurate, i.e.\ applying
the copoint to $\xi(x)$ actually returns the present state~$x$.

The main purpose of the information about the present state included
in~$T$ is to allow imposing relationships between the present point
and its collection of successors. Indeed, every functor~$T$ can be
made copointed by passing to the functor $T\times\id$ (given on
sets~$X$ by $(T\times\id)X=TX\times X$), with $\varepsilon(t,x)=x$; we
refer to copointed functors of this shape as \emph{trivially
  copointed}, as they impose no relationship between the present state
and its collection of successors. Copointed functors can absorb
non-iterative axioms; e.g.\ the modal logic~$\textbf{T}$ is captured
by the copointed functor $T$ given by
$TX=\{(A,x)\in\pow X\times X\mid x\in A\}$\label{page:t-functor} (more
details are given in Section~\ref{sec:zeroone-logic}), which imposes
that the present state is among its own successors; that is, proper
$T$-coalgebras are precisely reflexive Kripke frames. This functor~$T$
is our first example of a non-trivially copointed functor; note that
it is a subfunctor of the trivially copointed functor $\pow\times\id$.

\noindent For purposes of our strong completeness result, we make use
of a relaxed notion of copointed functor:

\begin{definition} A \emph{weakly copointed functor} $(T,\varepsilon)$
  (or just~$T$ when~$\varepsilon$ is clear from the context) consists
  of a functor~$T$ and a \emph{weak copoint}~$\varepsilon$, i.e.\ a
  natural transformation $\varepsilon: T \rightarrow \ultra$, where
  $\ultra$ denotes the (functor part of) the ultrafilter monad. That
  is, $\ultra X$ is the set of ultrafilters on~$X$, and
  $\ultra f(\alpha) = \{ B \subseteq Y \mid f^{-1}[B] \in \alpha\}$
  for $f\colon X\to Y$, $\alpha\in\ultra X$ (so $\ultra$ is a
  subfunctor of the neighbourhood functor~$\neighb$ as in
  Example~\ref{expl:logics}.\ref{nft}). Then, a $T$-coalgebra
  structure $\xi : X \rightarrow TX$ is \emph{proper} if
  $\varepsilon_X \circ \xi = \eta_X $ where
  $\eta: Id \rightarrow \ultra$ is the unit of the ultrafilter monad,
  given by $\eta_X(x) = \dot x=\{A \in \pow X \mid x \in A\}$. Every
  functor~$T$ induces a \emph{trivially weakly copointed functor}
  $T\times\ultra$, with second projection as the weak copoint.
\end{definition}
\noindent Instead of the identity of the present state, a weakly
copointed functor contains only a description of the present state,
which in general may fail to be realized as an actual state. However,
weakly copointed functors relate tightly to copointed functors in the
standard sense:

\begin{lemmadefn}\label{lem:copointed}
  Let $(T,\varepsilon)$ be a weakly copointed functor. Then
  \begin{equation*}
    T_cX=\{t\in TX\mid \varepsilon(t)\text{ principal}\}
  \end{equation*}
  defines a copointed subfunctor of~$T_c$, the \emph{copointed part}
  of~$T$, with copoint~$\varepsilon_c$ defined by
  $\varepsilon_c(t)\in\bigcap\varepsilon(t)$. Moreover, every proper
  $T$-coalgebra $C=(X,\xi)$ factors through the inclusion
  $T_cX\into TX$, inducing a coalgebra $C_c$ for the copointed
  functor~$T_c$. Given a similarity type~$\Lambda$ with assigned
  predicate liftings for~$T$, we obtain predicate liftings for~$T_c$
  by restriction; then, a state~$x\in X$ satisfies the same
  $\Lambda$-formulae in~$C$ as in~$C_c$.
\end{lemmadefn}
\noindent (Recall that an ultrafilter~$\alpha$ is \emph{principal} if
$\bigcap\alpha\neq\emptyset$, and then necessarily
$|\bigcap\alpha|=1$.)

\begin{remark}\label{rem:weak-copoints}
  Indeed, the above lemma implies that weakly copointed functors are
  not strictly required for our current target results, which are all
  formulated over proper coalgebras. We nevertheless do involve them
  in the technical development because of their natural role within
  coalgebraic logic: The 0-1-step logic, in the sense introduced in
  the next section, of the canonical $\Lambda$-structure, which will
  be based on a weakly copointed functor, is strongly complete; this
  would be impossible for any $\Lambda$-structure based on a copointed
  functor. For details, see
  Remark~\ref{rem:strong-01-step-completeness}.
\end{remark}

\begin{remark}\label{rem:one-to-zeroone}
  The standard coalgebraic semantics of rank-1 modal logics as
  recalled above embeds into the copointed setting by converting plain
  functors~$T$ into trivially copointed functors $T\times\id$ or
  trivially weakly copointed functors $T\times\ultra$, with modalities
  interpreted via first projections: Given a $\Lambda$-structure based
  on the functor~$T$, we obtain a $\Lambda$-structure based on the
  trivially copointed functor $T\times\id$ by putting
  \begin{equation*}
    (t,x)\models L(A_1,\dots,A_n)\quad\text{iff}\quad
    t\models L(A_1,\dots,A_n)
  \end{equation*}
  for $A_1,\dots,A_n\subseteq X$ and $(t,x)\in TX\times X$, similarly
  for $T\times\ultra$. Proper coalgebras for $T\times\id$ and proper
  coalgebras for $T\times\ultra$ are both essentially the same as
  (plain) coalgebras for~$T$, and it is easy to see that the
  respective modal semantics over $T$-coalgebras and over proper
  $T\times\id$- or $T\times\ultra$-coalgebras are equivalent.
\end{remark}

\begin{remark}
  The categorical concept of a \emph{comonad} extends the notion of
  copointed functor by additionally assuming an unfolding operation
  $\delta:T\to T\circ T$ (the \emph{comultiplication}) satisfying
  certain equational laws. This amounts to letting $T$ contain
  information about the entire finite-time future development of the
  present state: Iterating~$\delta$, we can extract evolutions of any
  depth~$n$, i.e.\ elements of $T^nX$, from a given element
  of~$TX$. Comonads can thus be employed to capture iterative frame
  conditions such as $\Box a\to\Box\Box a$, with the technical caveat
  that this requires restricting the branching degree of models to
  avoid set-theoretic existence problems. Since the meta-theory of
  iterative frame conditions is in general much less well-behaved than
  that of non-iterative ones (e.g.\ there are modal logics that are
  weakly complete but not strongly complete over neighbourhood
  semantics~\cite{shehtman1999strong}), one should manage expectations
  regarding the perspective of results in comparable generality as the
  present one.
\end{remark}
\section{The 0-1-Step Logic}\label{sec:zeroone-logic}
\noindent An important driving principle of coalgebraic logic is to
reduce metatheoretic properties of a full-blown modal logic with
nested modalities, interpreted over coalgebras, to similar properties
of a much simpler \emph{one-step logic} where formulae feature
precisely one layer of modalities, and are interpreted over structures
that essentially model just one transition step (hence the name). To
cover non-iterative logics, we need to extend this principle to cover
also the current state (besides its successors), arriving at the
\emph{0-1-step logic} of the given modal logic. For readability, we
\emph{restrict the technical development to unary modalities} from now
on; covering higher arities requires no more than additional indexing,
and we continue to use higher arities in the examples.

\paragraph{Syntax and derivations} In formulae of the 0-1-step logic,
we intentionally mix syntax and semantics, replacing propositional
variables by their values in a powerset Boolean algebra. That is,
given a non-iterative logic $\Lang = (\Lambda, \Axioms)$ and a
set~$X$, we take $\Prop(\Lambda(\pow X) \cup \pow X )$ to be the set
of \emph{0-1-step formulae over $\pow X$}, referring to elements
of~$\pow X$ as \emph{(interpreted) propositional atoms}. We denote the
evaluation of a $\Prop(\pow X)$-formula $\phi$ in the Boolean algebra
$\pow X$ by $\llbracket \phi \rrbracket$, and say that~$\phi$ is
\emph{propositionally valid over~$\pow X$} if $\Sem{\phi}=X$. We will
identify occurrences of subformulae~$\phi\in\Prop(\pow X)$ with
$\Sem{\phi}$ when they lie in scope of a modal operator but not
otherwise, i.e.\ on the uppermost level. This evaluation of inner
propositional formulae allows us to omit the modal congruence rule. We
thus define \emph{0-1-step derivability} $\nondash \psi$ of 0-1-step
formulae~$\psi$ inductively by the rules
\begin{gather*}
  \frac{\qquad}{\psi\sigma}\;(\psi \in \Axioms, \sigma \text{ a } \Prop(\pow  X)\text{-substitution})\\
  \frac{\phi_1, \ldots, \phi_n}{\psi}\;(\{\phi_1, \ldots, \phi_n\} \vdash_{PL} \psi)
  \qquad
  \frac{\qquad}{\phi}\;(\phi\in\Prop(\pow X), \Sem{\phi}=X).
\end{gather*}
(Non-iterative rules $\phi/\psi$ as in Remark~\ref{rem:rules}, if
present, are also applied in substituted form: if $\Sem{\phi\sigma}=X$
for a $\Prop(\pow X)$-substitution~$\sigma$, then derive
$\psi\sigma$.) That is, $\nondash \psi$ iff $\psi$ is
propositionally entailed by
\begin{equation*}
  \{\psi\sigma \mid \psi \in \Axioms, \sigma \text{ a }
  \Prop(\pow X)\text{-substitution}\}\cup\{\phi\mid\phi\in\Prop(\pow X),\Sem{\phi}=X\}.
\end{equation*}
We write $\Phi \nondash \psi$ if
$\nondash (\phi_1 \wedge \ldots \wedge \phi_n )\rightarrow \psi$ for
some $\phi_1, \ldots, \phi_n \in \Phi$.  A set~$\Phi$ of 0-1-step
formulae over $\pow X$ is \emph{0-1-step consistent} if
$\Phi \not \nondash \bot$.

\paragraph{Semantics} Fix a weakly copointed functor $(T,\varepsilon)$
and a $\Lambda$-structure~$\Struct$ based on~$T$. Define the unary predicate
lifting $\iota$ by
$\iota_X(A) = \{ t \in TX \mid A \in \varepsilon(t)\}$. The
\emph{0-1-step satisfaction} relation $t \models^{\zeroone}_X \psi$ between
functor elements $t \in TX$ and 0-1-step formulae~$\psi$ over~$\pow X$
is inductively defined by
  \begin{alignat*}{2}
    &t \not \models^{\zeroone}_X \bot &\\
    &t \models^{\zeroone}_X \neg\phi &&\text{ iff } t \not \models^{\zeroone}_X \phi \\
    &t \models^{\zeroone}_X \phi \wedge \psi &&\text{ iff } t
                                           \models^{\zeroone}_X \phi
                                           \text{ and } t \models^{\zeroone}_X \psi\\
    &t \models^{\zeroone}_X L\phi &&\text{ iff } t \in \llbracket L \rrbracket_X(\phi)\\
    &t \models^{\zeroone}_X B &&\text{ iff } t \in
                            \iota_X(B) 
  \end{alignat*}
  where $B\in\pow X$ in the last clause, and
  $\llbracket \psi \rrbracket^{\zeroone}_X = \{t \in TX \mid t
  \models^{\zeroone}_X \psi \}$.  The last clause thus deals with
  top-level interpreted propositional atoms.  Note that in accordance
  with the above convention, the second to last clause omits
  interpretation of modal arguments, which are already identified with
  their interpretation. We say that $\psi$ is \emph{satisfiable} if
  $\Sem{\psi}^{\zeroone}_X\neq\emptyset$, and we write
  $TX \models ^{\zeroone}_X \psi $ if
  $\llbracket \psi \rrbracket^{\zeroone}_X = TX$. We generally refer
  to maps $\tau\colon V\to\pow X$ as $\pow X$\emph{-valuations}. Given
  a non-iterative axiom $\psi$, we write $\psi\tau$ for the 0-1-step
  formula obtained from~$\psi$ by substituting according
  to~$\tau$. Then,~$\psi$ is \emph{0-1-step sound} for~$\Struct$ if
  $TX \models ^{\zeroone}_X \psi \tau$ for every set $X$ and every
  $\pow X$-valuation $\tau$.  Conversely, the
  logic~$\Lang=(\Lambda,\Axioms)$, or just~$\Axioms$, is
  \emph{0-1-step complete} for~$\Struct$ if every 0-1-step
  formula~$\psi$ over $\pow X$ such that
  $TX \models^{\zeroone}_X \psi$ is 0-1-step derivable
  ($\nondash \psi$), equivalently if every 0-1-step consistent formula
  is satisfiable. The same terminology applies to non-iterative rules
  (Remark~\ref{rem:rules}) (specifically, a non-iterative rule
  $\phi/\psi$ is 0-1-step sound if $TX\models ^{\zeroone}_X \psi \tau$
  whenever $\Sem{\phi\tau}=X$).

  To enable an appropriate statement of soundness, we extend the
  semantics of the logic to allow for \emph{frame conditions}: We
  refer to a pair $(C,\pi)$ consisting of a $T$-coalgebra~$C=(X,\xi)$
  and a valuation $\pi\colon V\to\pow X$ of the propositional
  variables as a \emph{$T$-model}. We define satisfaction
  $x\models_{(C,\pi)}\psi$ of 0-1-step formulae
  $\psi\in\Prop(\Lambda(\Prop(V)\cup V))$ in states~$x$ of $T$-models
  $(C,\pi)$ by the same clauses as for $\models_C$
  (Section~\ref{sec:semantics}), and additionally
  \begin{equation*}
    x\models_{(C,\pi)}a\quad\text{iff}\quad x\in\pi(a)
  \end{equation*}
  for $a\in V$. We say that $C$ satisfies the \emph{frame condition}
  $\psi$ if $x\models_{(C,\pi)}\psi$ for all $T$-models $(C,\pi)$. Of
  course, if~$C$ satisfies the frame condition $\psi$ then~$\psi$ is
  \emph{sound} for~$C$, i.e.\ every state in~$C$ satisfies all
  substitution instances of~$\psi$.

  \begin{lemma}[Soundness] \label{soundnessImplication} If a
    non-iterative axiom~$\psi$ over~$V$ is 0-1-step sound over a
    $\Lambda$-structure $\Struct$ based on a weakly copointed
    functor~$T$, then every proper $T$-coalgebra satisfies the frame
    condition~$\psi$; hence,~$\psi$ is sound for the class of all
    proper $T$-coalgebras.
\end{lemma}
\noindent We proceed to discuss in more detail how non-iterative
axioms are absorbed into (weakly) copointed functors. Given a (weakly)
copointed functor~$T$ and a set~$\Axioms'$ of additional non-iterative
axioms, we can pass to the (weakly) copointed
subfunctor~$T_{\Axioms'}$ of~$T$ given by
\begin{equation*}
  T_{\Axioms'} X=\{ t\in T\mid t\models^{\zeroone}_X\phi\sigma\text{ for all $\phi\in\Axioms'$ and all $\pow X$-substitutions $\sigma$}\}
\end{equation*}
and restrict the $\Lambda$-structure to $T_{\Axioms'}$ in the evident
way. By construction, the axioms in~$\Axioms'$ are 0-1-step sound
over~$T_{\Axioms'}$, and the proper $T_{\Axioms'}$-coalgebras are
precisely those proper $T$-coalgebras that satisfy the axioms
in~$\Axioms'$ as frame conditions. Moreover, we have
\begin{lemma}\label{lem:axiom-completeness}
  In the notation introduced above, suppose that the set~$\Axioms$ of
  non-iterative axioms is 0-1 step sound and 0-1 step complete
  over~$T$. If $\Axioms'$ mentions only finitely many modalities, then
  $\Axioms\cup\Axioms'$ is 0-1-step complete over~$T_{\Axioms'}$.
\end{lemma}
\begin{proof*}{Proof (sketch).}
  Observe that if $\psi$ is a 0-1-step formula over~$\pow X$ such that
  $T_{\Axioms'}X\models^{\zeroone}_X\psi$, with~$X$ assumed to be
  finite w.l.o.g., then $T X\models^{\zeroone}_X(\Land\Phi)\to\psi$
  where~$\Phi$ contains representatives up to propositional
  equivalence of all instances of axioms in~$\Axioms'$ under
  $\pow X$-substitutions; the assumptions guarantee that we can
  take~$\Phi$ to be finite.
\end{proof*}
\begin{example}
  \begin{enumerate}[wide]
  \item We have recalled the coalgebraic view on standard Kripke
    semantics in Example~\ref{expl:logics}.\ref{item:Kripke}. The
    usual axioms of the modal logic~$K$ ($\Box\top$ and
    $\Box(a\to b)\to\Box a\to\Box b)$ are 0-1-step complete over the
    trivially copointed functor $\pow\times\id$ induced by the
    functor~$\pow$; this is implied by translating the known
    \emph{one-step} completeness of these axioms
    over~$\pow$~\cite{DBLP:journals/tcs/Pattinson03} into the
    copointed setting as indicated in
    Remark~\ref{rem:one-to-zeroone}. It follows by
    Lemma~\ref{lem:axiom-completeness} that these axioms, together
    with the $T$-axiom $\Box a\to a$, are 0-1-step complete for the
    copointed functor~$T$ given by
    \begin{equation*}
      TX=\{(B,x)\in\pow X\times X\mid (B,x)\models\Box A\to A\text{ for all $A\in\pow X$}\}.
    \end{equation*}
    It is easy to see that $TX=\{B,x)\in\pow X\times X\mid x\in B\}$,
    i.e.~$T$ coincides with the copointed functor recalled on
    p.~\pageref{page:t-functor}, whose proper coalgebras are the
    reflexive Kripke frames.
  \item The assumption that the additional axioms only mention
    finitely many modalities is really needed; without it, the claim
    fails even in the rank-1 case. For instance, let~$\mathcal{S}$ be
    the \emph{subdistribution functor}, which assigns to a set~$X$ the
    set $\mathcal{S}X$ of discrete subdistributions on~$X$, where a
    subdistribution is defined like a distribution except that the
    weight of the whole set is required to be at most~$1$ rather than
    equal to~$1$. We use modalities $L_p$ `with weight at least~$p$'
    with the same semantics as in the probabilistic case
    (Example~\ref{expl:logics}.\ref{item:prob}). Take the set 
    \begin{equation*}
      \Axioms'=\{\neg L_1\top\}\cup\{ L_{1-1/n}\top\mid n\ge 1\}
    \end{equation*}
    of rank-1 axioms. Then
    $(\mathcal{S}\times\id)_{\Axioms'}X=\emptyset$ for all~$X$, so
    that $(\mathcal{S}\times\id)_{\Axioms'}X\models^\zeroone_X\bot$,
    but $\bot$ is not derivable under the given axioms (together with
    any sound axiomatization of~$\mathcal{S}$), as any derivation
    of~$\bot$ could only use a finite subset of~$\Axioms'$, and all
    such finite subsets are clearly consistent.
  \end{enumerate}
\end{example}
\noindent A key role in the completeness proof will be played by the
following subformula property of the 0-1-step logic, which extends
\cite[Proposition~24]{DBLP:journals/logcom/SchroderP10} from rank-1 to
non-iterative logics.

\begin{proposition} \label{sfp} Let $\psi$ be a 0-1-step formula
  over~$\pow X$ such that $\nondash \psi $. Then~$\psi$ is 0-1-step
  derivable using only $\Prop(\mathfrak{A})$-instances of axioms and
  $\Prop(\mathfrak A)$-formulae valid over~$\pow X$, where
  $\mathfrak{A} \subseteq \pow X$ are the sets occurring in~$\psi$.
\end{proposition}

\noindent The proof requires some facts about propositional logic.

\begin{lemma} \label{boolEq} Let $V$ and $W$ be disjoint finite
  sets. For an $A$-valuation $\tau$ on $V$ with $A \subseteq \pow X$
  and a system of Boolean equations $\phi_i\tau = \psi_i\tau$ for
  $i = 1, \ldots, n$ where $\phi_i, \psi_i \in \Prop(V \cup W)$, if
  there exists an $A$-valuation $\kappa$ for $W$ such that
  $\phi_i\tau\kappa = \psi_i\tau\kappa$ for $i = 1, \ldots, n$, then
  there exists a $\Prop(V)$-substitution $\sigma$ on $W$ such
  that \begin{enumerate}
  \item \label{item:ba-solve} $\phi_i\sigma\tau = \psi_i\sigma\tau$ for
    $i = 1, \ldots, n$
  \item \label{item:monot-solve}$x\kappa \subseteq \llbracket x\sigma\tau \rrbracket$ for
    $x \in W$ if $\vert W \vert = 1$.
  \end{enumerate}
\end{lemma}
\noindent (Claim~(\ref{item:ba-solve}) says effectively that if
Boolean equations with coefficients in~$A$ are solvable in~$A$, then
they are solvable by Boolean combinations of the coefficients that
actually occur. Claim~(\ref{item:monot-solve}) is only needed later.)
\begin{proof} \emph{(\ref{item:ba-solve}):} This is well-known but we
  need the construction for Claim~(\ref{item:monot-solve}). We
  immediately reduce to a single equation $\phi\tau = \top$ where
  \begin{math}
    \phi =\textstyle \bigwedge_{i=1}^n (\phi_i \leftrightarrow \psi_i).
  \end{math}
  We construct $\sigma$ by induction over $\vert W\vert$, with trivial
  base $\vert W \vert = 0$.  In the inductive step, we pick $x \in W$
  and obtain, by Boolean expansion,
  \begin{align*} \phi &\equiv (x \rightarrow \phi[\top / x]) \wedge
    (\neg x \rightarrow \phi[\bot / x])\\ &\equiv (x \rightarrow
    \phi[\top / x]) \wedge (\neg \phi[\bot / x] \rightarrow x),
  \end{align*}
  which in turn entails
  $\neg \phi[\bot / x] \rightarrow \phi[\top / x]$, so by assumption
  the equation
  $(\neg \phi[\bot / x] \rightarrow \phi[\top / x])\tau = \top$ over
  $W \setminus \{x\}$ is solved by $\kappa$, and hence by induction
  solvable by some $\Prop(V)$-substitution $\sigma'$.  Thus, the
  substitution
  \begin{equation*}
    \sigma = [\phi[\top / x] / x ]\sigma'
  \end{equation*}
  for $W$ satisfies $\phi\sigma\tau = \top$.
  
  \emph{(\ref{item:monot-solve}):} Let $W = \{x\}$; we then have
  constructed $\sigma = [\phi[\top / x] / x ]$
  in~(\ref{item:ba-solve}). We have to show
  $\kappa(x) \subseteq \llbracket \phi[\top/x] \tau \rrbracket$. Let
  $y \in \kappa(x)$ and assume w.l.o.g.\ that $\phi$ is in CNF, and
  that~$x$ appears in at most one literal in every clause $\psi$ in
  $\phi$. We have to show that
  $y \in \llbracket \psi[ \top / x ]\tau \rrbracket $. If the
  literal~$x$ appears in $\psi$, then this holds trivially. Otherwise,
  $\psi$ must contain some literal not mentioning~$x$ whose
  interpretation contains~$y$, since $\psi \tau \kappa = \top$ by
  assumption and
  $y\notin\Sem{(\neg x)\kappa}=\Sem{(\neg x)\tau\kappa}$.  Therefore
  $y \in \llbracket \phi[\top/x] \tau \rrbracket$ as
  required.  \end{proof}

\begin{lemma} \label{weakSubstLemma} Let $\Phi \subseteq \Prop(V)$,
  let $\psi \in \Prop(V)$, and let $\sigma$ be a $W$-substitution on
  $V$ and $\tau$ a $U$-substitution on $V$ such that
  $\tau(a) = \tau(b)$ whenever $\sigma(a) = \sigma(b)$ for all
  $a, b \in V$, and moreover $\Phi\sigma \vdash_{PL} \psi\sigma$. Then
  $\Phi\tau \vdash_{PL} \psi\tau$.
\end{lemma}

\begin{lemma} \label{crr} Let $\Phi\subseteq\Prop(V)$, and
  let~$\psi\in\Prop(V)$. Given a $U$-substitution~$\sigma$ and a
  $W$-substitution~$\tau$ on~$V$, if
  $\Phi\sigma \vdash_{PL} \psi\sigma$ then
  $\Phi\tau \cup \Psi \vdash_{PL} \psi\tau$, where
  $\Psi = \{ \tau(a) \leftrightarrow \tau(b) \mid a, b \in V,
  \sigma(a) = \sigma(b) \}$.
\end{lemma}

\begin{lemma} \label{substLemma} Let $V$ and $W$ be disjoint sets, let
  $W_0\subseteq W$, let $\Phi \subseteq \Prop(V)$, let
  $\psi \in \Prop(W_0)$, and let $\sigma$ and $\tau$ be
  $W$-substitutions on $V$ such that $\tau(a) = \tau(b)$ whenever
  $\sigma(a) = \sigma(b)$ and $\tau(a) = c$ whenever $\sigma(a) = c$
  for all $a, b \in V$ and $c \in W_0$, and moreover
  $\Phi\sigma \vdash_{PL} \psi$. Then $\Phi\tau \vdash_{PL}
  \psi$.
\end{lemma}
\begin{proof} Let $\sigma'$ and $\tau'$ be the $W$-substitutions on
  $V \cup W_0$ such that $\sigma'(w) = \tau'(w) = w$ for $w \in W_0$ and
  $\sigma'(v) = \sigma(v)$, $\tau'(v) = \tau(v)$ for $v\in V$.  The
  claim then follows by Lemma \ref{weakSubstLemma}.
\end{proof}

\begin{proof*}{Proof of Proposition~\ref{sfp}.}
  Let $V$ be a sufficiently large set of
  propositional variables. Then there are finite sets $\Phi_1$ of
  $\Prop(V)$-instances of axioms and~$\Phi_2\subseteq\Prop(V)$ that we
  can assume to be instantiated by a single $\pow X$-valuation
  $\sigma$ such that the formulae in~$\Phi_2\sigma$ are
  propositionally valid over~$\pow X$ and
  $(\Phi_1\cup\Phi_2)\sigma \vdash _{\textup{PL}} \psi$. By Lemma
  \ref{substLemma}, it suffices to show that there is a
  $\Prop(\mathfrak{A})$-substitution $\tau$ that solves the following
  system of equations:

  \begin{itemize} \item For all subformulae $L\rho, L\rho'$ in
    $\Phi_1$ such that $(L \rho)\sigma=(L\rho')\sigma$ in
    $\Lambda(\pow X)$, we have $(L \rho)\tau=(L\rho')\tau$ in
    $\Lambda(\pow X)$.  This amounts to an equation $\rho = \rho'$.
    
  \item For all subformula $LA$ in $\psi$ and $L \rho$ in $\Phi_1$
    such that $LA=(L \rho)\sigma$ in $\Lambda(\pow X)$, we have
    $LA=(L \rho)\tau$ in $\Lambda(\pow X)$. This amounts to an
    equation $A = \rho$.
    
  \item For all subformulae $\rho,\rho'$ in $\Phi_1\cup\Phi_2$ that do
    not lie beneath a modal operator and are such that
    $\rho\sigma=\rho'\sigma$, we have $\rho\tau=\rho'\tau$ in
    $\pow X$. This amounts to an equation $\rho = \rho'$.
    
  \item For all subformulae $A$ in $\psi$ and $\rho$ in
    $\Phi_1\cup\Phi_2$ that do not lie beneath a modal operator and
    are such that $\rho\sigma=A$ in $\pow X$, we have $\rho\tau=A$ in
    $\pow X$.  This amounts to an equation $A = \rho$.

  \end{itemize}
  By construction, this system of Boolean equations is solvable by
  $\sigma$, and since only sets from $\mathfrak{A}$ appear in the
  equations, by Lemma \ref{boolEq}.(\ref{item:ba-solve}) it is also
  solvable by a $\Prop(\mathfrak{A})$-substitution with the required
  properties.
\end{proof*}

\section{Copointed Coalgebraic Logics are
  Non-Iterative}\label{sec:fmp}

We next establish that weakly copointed functors are indeed
characterized by non-iterative axioms; that is, we \emph{fix for this
  section a $\Lambda$-structure~$\Struct$ based on a weakly copointed
  functor~$T$} and show that there is a set of non-iterative axioms
that is sound and weakly complete over the class of all proper
$T$-coalgebras. (We necessarily restrict to weak completeness, since
coalgebraic modal logics in general fail to be
compact~\cite{DBLP:journals/jlp/Schroder07}). In more detail, we show
that 0-1-step completeness of a non-iterative axiomatization implies
its weak completeness over finite models, and we show that the set of
all 0-1-step sound non-iterative axioms is 0-1-step complete. The
proofs are fairly straightforward generalizations of the rank-1
case~\cite{DBLP:journals/jlp/Schroder07}. We begin with the latter
step:

\begin{theorem}
	\label{converse}
	The set of all 0-1-step sound 0-1-step axioms is 0-1-step complete.
\end{theorem}
\begin{proof}
  By Remark~\ref{rem:rules}, it suffices to show that the set of all
  0-1-step sound non-iterative rules is 0-1-step complete. Let
  $TX\models^{\zeroone}_X \psi$ for a 0-1-step formula~$\psi$
  over~$\pow X$. Then~$\psi$ has the form $\psi=\psi_0\tau$ for
  $\psi_0\in\Prop(\Lambda(V_0)\cup V_0)$, with $V_0\subseteq V$
  finite, and a $\pow X$-valuation $\tau$. Let $\phi$ be the
  conjunction of all clauses~$\chi$ over~$V_0$ such that
  $\Sem{\chi \tau}=X$; then $\Sem{\phi\tau}=X$. We are thus done once
  we show that $\phi / \psi_0$ is 0-1-step sound.  So assume
  $\Sem{\phi \sigma}=Y$ for a $\pow Y$-valuation $\sigma$. We have to
  show $TY \models^{\zeroone}_Y \psi_0 \sigma$. For each $y \in Y$
  there is $x \in X$ such that for all $a \in V_0$ we have
  $x \in \tau(a)$ iff $y \in \sigma(a)$ (otherwise there is a
  clause~$\chi$ over~$V_0$ such that $X \models \chi\tau$ but
  $Y \not \models \chi \sigma$, contradicting
  $Y \models \phi \sigma$). Therefore there is $f: Y \rightarrow X$
  such that $\sigma(a) = f^{-1}[\tau(a)]$ for all $a \in V_0$. By
  naturality of predicate liftings (including~$\iota$) and commutation
  of preimage with all Boolean operations, we have
  $\llbracket \psi_0 \sigma \rrbracket_Y^{\zeroone} =
  Tf^{-1}[\llbracket \psi_0 \tau \rrbracket_X^{\zeroone}]$, and
  therefore $TY \models^{\zeroone}_Y \psi_0 \sigma$ as required.
\end{proof}
\noindent We will base all our model constructions on the following
central notions:
\begin{definition}
  A set~$\Sigma$ of formulae is \emph{closed} if it is closed under
  subformulae and negations of formulae that are not themselves
  negations. We write $C_\Sigma$ for the set of maximally consistent
  subsets of~$\Sigma$.  For a $\Lambda$-formula~$\phi$, we write
$\hat{\phi}=\{\Phi\in C_\Sigma \mid \phi\in\Phi\}$.
\end{definition}
\begin{lemma}\cite[Lemma~27]{DBLP:journals/jlp/Schroder07} \label{allDerivable}
  Let $\phi$ be a propositional formula over $V$, $\sigma$ a
  $\Sigma$-substitution and $\hat{\sigma}$ a
  $\pow(C_\Sigma)$-valuation with $\hat{\sigma}(a) = \hat{\psi}$
  when $\sigma(a) = \psi$. Then
  $ \llbracket \phi\hat{\sigma} \rrbracket = C_\Sigma$ iff
  $\ldash \phi\sigma$.
\end{lemma}
\begin{definition}
  Let~$\Sigma$ be closed. A coalgebra $(C_\Sigma, \xi)$ is
  \emph{coherent} if for all $L\psi\in\Sigma$, $\Phi\in C_\Sigma$,
  \begin{equation*}
    \xi(\Phi) \in \llbracket L \rrbracket_{C_\Sigma}(\hat{\psi})
    \quad\text{iff}\quad L\psi \in \Phi.
  \end{equation*}
\end{definition}
\begin{lemma}[Truth
  lemma~\cite{DBLP:journals/jlp/Schroder07}]\label{truthLemma}
  Let $\Sigma$ be closed, and let $C = (C_\Sigma, \xi)$ be a coherent
  $T$-coalgebra and let $\phi \in \Sigma$. For all $\phi \in \Sigma$ we then
  have $\Phi \models_C \phi$ iff $\phi \in \Phi$.
\end{lemma}
\noindent Thus, model constructions reduce to showing the existence of
coherent coalgebra structures. The latter requires the following
lemma, which for later reuse we prove for possibly infinite~$\Sigma$:


\begin{lemma} \label{ci} Let $V_\Sigma$ denote the set
  $\{a_\phi \mid \phi \in \Sigma\}$, and let
  $\Phi\subseteq\Prop(\Lambda(V_\Sigma) \cup V_\Sigma)$. Let
  $\sigma$ be the substitution given by $\sigma(a_\phi) = \phi$, and
  let $\hat{\sigma}$ be the $\pow C_\Sigma$-valuation given by
  $\hat{\sigma}(a_\phi) = \hat{\phi}$. If $\Phi \sigma$ is consistent,
  then $\Phi\hat{\sigma}$ is 0-1-step consistent.
\end{lemma}

\begin{proof} By contraposition; so assume
  $\Phi \hat{\sigma} \vdash^{\zeroone}_\Sigma \bot$. By
  Proposition~\ref{sfp}, there is a derivation that uses only
  $\Prop(\mathfrak{A})$-instances of axioms and
  $\Prop(\mathfrak{A})$-formulae valid over $\pow C_\Sigma$, for
  $\mathfrak{A} = \{\hat{\phi} \mid \phi \in \form(\Lambda)\}$. We can
  write the set of these formulae as $\Theta\hat\sigma$ for a
  set~$\Theta\subseteq\Prop(\Lambda(V_\Sigma) \cup V_\Sigma)$. By
  the definition of 0-1-step derivations, it follows that
  $(\Phi \cup \Theta)\hat{\sigma} \vdash_{PL} \bot$. Now let $\Psi$
  denote the set
  $\{L \rho\leftrightarrow L\rho' \mid \hat{\rho} =
  \hat{\rho}'\}$. The formulae in~$\Psi$ are derivable in $\Lang$ by
  Lemma~\ref{allDerivable} and the congruence rule. Similarly, let
  $\Gamma = \{ \phi \leftrightarrow \phi' \mid \hat{\phi} =
  \hat{\phi}' \}$; the formulae in~$\Gamma$ are $\Lang$-derivable by
  Lemma~\ref{allDerivable}.  By Lemma~\ref{crr}, it follows that
  $(\Phi \cup \Theta)\sigma \cup \Psi \cup \Gamma \vdash_{PL} \bot$
  and therefore (again using Lemma~\ref{allDerivable})
  $\Phi\sigma \ldash \bot$.
\end{proof}

\begin{lemma}[Finite existence lemma]\label{lem:ex-finite}
  Let $\Axioms$ be 0-1-step complete, and let~$\Sigma$ be a finite
  closed set of formulae.  Then there exists a coherent proper
  $T$-coalgebra structure $\xi$ on $C_\Sigma$.
\end{lemma}
\begin{proof}
  Let $\Phi \in C_\Sigma$. We show that the requirements on
  $\xi(\Phi)$ form a 0-1-step consistent 0-1-step formula, implying
  existence of~$\xi(\Phi)$ by 0-1-step completeness. Take $V_\Sigma$,
  $\sigma$ and~$\hat\sigma$ as in Lemma~\ref{ci}. Let
  \begin{equation*}
    \chi=\textstyle\bigwedge_{L \psi \in \Phi} L a_{\psi}\land
    \bigwedge_{\neg L \psi \in \Phi}\neg L a_\psi\land \bigwedge_{\psi \in \Phi} a_\psi.
  \end{equation*}
  We need to show that $\chi\hat\sigma$ is 0-1-step consistent.  By
  Lemma~\ref{ci}, this follows from consistency of $\chi\sigma$, which
  in turn is implied by consistency of~$\Phi$. 
\end{proof}
\noindent The announced weak completeness result now follows:

\begin{theorem}[Weak completeness and bounded model
  property]\label{thm:bmp}
  Let $\Axioms$ be 0-1-step complete for the $\Lambda$-structure
  $\Struct$. Then $\Axioms$ is weakly complete over finite proper
  $T$-coalgebras; specifically, every consistent formula~$\phi$ is
  satisfiable in a finite proper $T$-coalgebra of size at
  most~$2^{|\phi|}$.
\end{theorem}

\begin{proof}
  Let~$\Sigma$ be the smallest closed set containing~$\phi$. By the
  finite existence lemma (Lemma~\ref{lem:ex-finite}), there is a
  proper and coherent $T$-coalgebra $\xi$ on $C_\Sigma$; note
  $|C_\Sigma|\le2^{|\phi|}$. Since $\Sigma$ has only finitely many
  consistent subsets, the consistent set $\{\phi\}$ is contained in
  some $\Phi\in C_\Sigma$. By the truth lemma,
  $\Phi \models_{(C_\Sigma, \xi)} \phi$.
\end{proof}

\begin{remark}\label{rem:fmp-fap}
  Previous work on the connection between algebraic and coalgebraic
  semantics~\cite{DBLP:conf/fossacs/PattinsonS08} has led to results
  that in particular cover non-iterative frame conditions. The
  technical setup in the mentioned work features an underlying rank-1
  logic, equipped with standard coalgebraic semantics using plain
  functors, and imposes additional frame conditions as axioms, e.g.\
  non-iterative frame conditions. One of the results obtained
  \cite[Corollary~37]{DBLP:conf/fossacs/PattinsonS08} shows that a
  coalgebraic logic with non-iterative frame conditions is weakly
  complete over coalgebras satisfying the frame conditions, provided
  that the frame conditions mention only finitely many modalities. By
  Remark~\ref{rem:one-to-zeroone} and
  Lemma~\ref{lem:axiom-completeness}, these assumptions allow
  combining the given rank-1 logic and the additional frame conditions
  into a 0-1-step complete logic for the copointed functor defined by
  the axioms. The weak completeness result therefore follows also from
  our Theorem~\ref{thm:bmp}, which moreover applies also to sets of
  non-iterative frame conditions that mention infinitely many
  modalities; of course, 0-1-step completeness then needs to be proved
  without the help of Lemma~\ref{lem:axiom-completeness}. E.g.\ this
  will turn out to be possible for the canonical $\Lambda$-structure
  introduced next (Lemma~\ref{lem:can-zeroone-complete}). 
\end{remark}

\section{The Canonical $\Lambda$-Structure}\label{sec:can-struct}

We now construct, \emph{for a given non-iterative logic
  $\Lang=(\Lambda,\Axioms)$ that we fix from now on}, a canonical
$\Lambda$-structure~$\Struct_\Lang$ based on a weakly copointed
functor~$M_\Lang$ w.r.t.\ which we show soundness and strong
completeness by means of a canonical model construction.  As usual,
the state space of the canonical model will be the set of maximally
consistent sets, denoted $C_\Lang$ (so $C_\Lang=C_{\form(\Lambda)}$ in
the notation of Section~\ref{sec:fmp}).

We construct the functor~$M_\Lang$ as follows. For a set~$X$,
$M_\Lang X$ is the set of maximally 0-1-step consistent subsets of
$\Prop(\Lambda(\pow X) \cup \pow X)$ (i.e.\ of the set of 0-1-step
formulae over $\pow X$). For a function $f\colon X \rightarrow Y$, we
define $M_\Lang f$ by
\begin{equation*}
M_\Lang f(\Phi) = \{ \phi \in \Prop(\Lambda(\pow Y) \cup \pow Y)
\mid \phi\sigma_f \in \Phi\}
\end{equation*}
where $\sigma_f$ is the $\pow X$-substitution on $\pow Y$ given by
$\sigma_f(A) = f^{-1}[A]$. We define a weak copoint
$\varepsilon\colon M_\Lang \rightarrow \ultra$ by
$\varepsilon_X(\Phi) = \Phi \cap \pow X $ for $\Phi \in M_\Lang X$,
and interpret $L \in \Lambda$ by
\begin{equation*}
  \llbracket L \rrbracket_X A = \{\Phi \in M_\Lang X
  \mid L A \in \Phi \}\qquad\text{for $A\subseteq X$.}
\end{equation*}
\noindent Of course, we intend an element of $M_\Lang X$ to satisfy
precisely the 0-1-step formulae that it contains; indeed, we have
\begin{lemma}[0-1-step truth lemma] \label{01stepTruthLemma} Let
  $\psi$ be a 0-1-step formula over $\pow X$. Then
  $\Phi \models^{\zeroone}_X \psi $ iff $ \psi \in \Phi$, for
  $\Phi \in M_\Lang X$.
\end{lemma}
\noindent Since a maximally consistent set in $M_\Lang X$ must in
particular contain all $\pow X$-instances of the axioms in~$\Axioms$,
it follows that $\Axioms$ is 0-1-step sound, and hence sound by Lemma
\ref{soundnessImplication}, for~$\Struct_\Lang$.

With a view to proving also 0-1-step completeness, we note a 0-1-step
version of the well-known Lindenbaum lemma:
\begin{lemma}[0-1-step Lindenbaum lemma]
  \label{lem:zeroone-lindenbaum} Every 0-1-step consistent set of
  0-1-step formulae over $\pow X$ is contained in a maximal such set.
\end{lemma}

\noindent From the 0-1-step truth lemma and the 0-1-step Lindenbaum
lemma, 0-1-step completeness is immediate:
\begin{lemma}\label{lem:can-zeroone-complete}
  The logic~$\Lang$ is 0-1-step complete for~$\Struct_\Lang$.
\end{lemma}
\noindent By Theorem~\ref{thm:bmp}, this implies weak completeness and
the finite (in fact, bounded) model property:
\begin{corollary}\label{cor:can-weak-compl}
  The logic~$\Lang$ is weakly complete over finite proper
  $M_\Lang$-coalgebras.
\end{corollary}
\noindent Our main result, established in the next section, will show
that~$\Lang$ is in fact \emph{strongly} complete over proper
$M_\Lang$-coalgebras (of course, one can then no longer restrict to
finite coalgebras). As indicated in the introduction, the canonical
$\Lambda$-structure is essentially neighbourhood semantics. We proceed
to elaborate details.

Recall from Example~\ref{expl:logics}.\ref{nft} that the
$\Lambda$-neighbourhood functor~$\neighb_\Lambda$ is defined as
$\neighb_\Lambda=\textstyle\prod_{L\in \Lambda\text{ $n$-ary}}
\contrapow \circ ((\contrapow^\Op)^n)$. Recall that $\neighb_\Lang$
induces a weakly copointed functor $\neighb_\Lambda\times\ultra$. Take
$\neighb_\Lang$ to be the weakly copointed subfunctor of
$\neighb_\Lambda\times\ultra$ defined by the the axioms~$\Axioms$,
i.e.\ 
\begin{equation*}
  \neighb_\Lang = (\neighb_\Lambda\times\ultra)_\Axioms
\end{equation*}
in notation introduced in Section~\ref{sec:zeroone-logic}.  It is
straightforward to see that the proper $\neighb_\Lang$-coalgebras are
precisely the $\Lambda$-neighbourhood frames satisfying the frame
conditions~$\Axioms$. The functors $\neighb_\Lang$ and $M_\Lang$ are
naturally isomorphic via the transformation
$\theta\colon M_\Lang\to \neighb_\Lang$ given by
\begin{equation*}
  \theta_X(\Phi)_L=(\{A\subseteq X\mid LA\in\Phi\},\{A\subseteq X\mid
  A\in\Phi\}),
\end{equation*}
which is also compatible with the predicate liftings. We can thus
translate Corollary~\ref{cor:can-weak-compl} into the language of
neighbourhood semantics:
\begin{corollary}\label{cor:nbhd-weak-compl}
  The logic~$\Lang=(\Lambda,\Axioms)$ is weakly complete over the
  class of finite neighbourhood frames that satisfy the axioms
  in~$\Axioms$ as frame conditions.
\end{corollary}
\noindent That is, one instance of the coalgebraic weak completeness
theorem (Theorem~\ref{thm:bmp}) is weak completeness of non-iterative
modal logics over their neighbourhood semantics as originally proved
by Lewis~\cite{DBLP:journals/jphil/Lewis74}.

\begin{remark}
  The weak completeness result in the above-mentioned previous work on
  algebraic-coalgebraic
  semantics~\cite[Corollary~37]{DBLP:conf/fossacs/PattinsonS08} (see
  Remark~\ref{rem:fmp-fap}) similarly puts weak neighbourhood
  completeness of non-iterative logics in a coalgebraic context: Given
  a rank-1 logic~$\Lang$, the canonical $\Lambda$-structure for the
  given rank-1 logic satisfies the conditions
  of~\cite[Corollary~37]{DBLP:conf/fossacs/PattinsonS08}, in
  particular is \emph{one-step complete} (the simpler version of
  0-1-step completeness that applies to rank-1
  logics)~\cite{DBLP:journals/logcom/SchroderP10}, and is isomorphic
  to the subfunctor of the neighbourhood functor defined by the given
  rank-1 axioms; \cite[Corollary~37]{DBLP:conf/fossacs/PattinsonS08}
  then guarantees that weak completeness is retained in any extension
  of~$\Lang$ with non-iterative axioms mentioning only finitely many
  modalities. By comparison, Corollary~\ref{cor:nbhd-weak-compl} above
  removes the restriction to finitely many modalities.
\end{remark}

\begin{tremark}[Strong 0-1-step
  completeness] \label{rem:strong-01-step-completeness} The strong
  completeness proof for rank-1 canonical
  structures~\cite{DBLP:journals/logcom/SchroderP10} (which implies
  the known result that every rank-1 logic is strongly complete over
  its neighbourhood semantics~\cite{Surendonk97}) can be factored
  through establishing \emph{strong one-step completeness}, i.e.\
  showing that the one-step logic (the simpler version of the 0-1-step
  logic that suffices in the rank-1 case) of a canonical structure is
  strongly
  complete~\cite[Remark~55]{DBLP:journals/logcom/SchroderP10}. Similarly,
  the 0-1-step logic of the canonical
  $\Lambda$-structure~$\Struct_\Lang$ defined above is strongly
  complete; that is, for every set~$X$, every consistent set of
  0-1-step formulae over~$\pow X$ is satisfiable
  over~$\Struct_\Lang$. Indeed, this is immediate from the 0-1-step
  truth lemma (Lemma~\ref{01stepTruthLemma}) and the 0-1-step
  Lindenbaum lemma (Lemma~\ref{lem:zeroone-lindenbaum}). On the other
  hand, the 0-1-step logic of the copointed part of the canonical
  $\Lambda$-structure, or indeed of any copointed functor, clearly
  fails to be strongly complete: Let $\alpha$ be a non-principal
  ultrafilter on a set~$X$; then $\alpha$ can be seen as a set of
  0-1-step formulae over~$\pow X$, and as such is consistent;
  but~$\alpha$ is clearly not satisfiable over any copointed
  functor. Strong completeness of the 0-1-step logic is the moral
  reason we include weakly copointed functors in the technical
  development even though, as indicated in
  Remark~\ref{rem:weak-copoints}, we could in principle short-circuit
  them.
\end{tremark}

\section{Strong Completeness}\label{sec:str-compl}
We proceed to prove our main result, strong completeness of
non-iterative modal logics over their canonical structure, to which
the known strong completeness over neighbourhood
semantics~\cite{Surendonk97} is a corollary. The centrepiece of the
technical development is an existence lemma; we set out to prepare its
proof. As usual, one has

\begin{lemma}[Lindenbaum Lemma]\label{lem:lindenbaum} Every consistent
  set of $\Lambda$-formulae is contained in a maximally consistent
  set.
\end{lemma}

\noindent The existence lemma requires us to show 0-1-step consistency
of a set of 0-1-step formulae specifying coherence and properness. We
start with the following observation, which is fairly immediate by
Lemma~\ref{ci}:

\begin{lemma} \label{consistency1} Let $\Phi \in C_\Lang$ be a maximally
  consistent set. Then the set
  \begin{equation*}
    \{ L\hat{\phi} \mid L\phi \in \Phi \}
    \;\cup\; \{ \neg L\hat{\phi} \mid \neg L\phi \in \Phi \} \;\cup\; \{
    \hat{\phi} \mid \phi \in \Phi \}
  \end{equation*}
  of 0-1-step formulae over $\pow C_\Lang$ is 0-1-step consistent.
\end{lemma}
\noindent The key step is then to extend the last component of the
union above from expressible subsets of~$C_\Lang$ to arbitrary
subsets:
\begin{lemma} \label{consistency2} Let $\Phi \in C_\Lang$ be a maximally
  consistent set. Then the set
  \begin{equation*} \{
    L\hat{\phi} \mid L\phi \in \Phi \} \;\cup\;  \{ \neg L\hat{\phi} \mid \neg L\phi \in \Phi \} \;\cup\;
    \dot{\Phi}
  \end{equation*}
  of 0-1-step formulae over $\pow C_\Lang$ is 0-1-step consistent.
\end{lemma}
\noindent Recall here that
$\dot\Phi=\{A\subseteq C_\Lang\mid \Phi\in A\}$ is the principal
ultrafilter generated by~$\Phi$, and note
$\dot\Phi\supseteq\{ \hat{\phi} \mid \phi \in \Phi \}$.  The proof
makes central use of Lemma~\ref{boolEq}.(\ref{item:ba-solve})
and~(\ref{item:monot-solve}) in a step-wise elimination of atoms
in~$\dot\Phi\setminus\{ \hat{\phi} \mid \phi \in \Phi \}$ from
0-1-step derivations. With Lemma~\ref{consistency2} in place, the
existence lemma follows straightforwardly:

\begin{lemma}[Existence lemma]\label{existenceLemma} There exists a
  coherent proper $M_{\Lang}$-coalgebra on $C_{\Lang}$.  \end{lemma}
\noindent Using the Lindenbaum lemma~\ref{lem:lindenbaum} and the
truth lemma (Lemma~\ref{truthLemma}) in the standard fashion, we then
obtain our main result, strong completeness over the canonical
coalgebraic semantics:
\begin{theorem}[Coalgebraic strong
  completeness] \label{StrongCompletenes} The logic $\Lang$ is
  strongly complete over proper $M_\Lang$-coalgebras, and hence over
  coalgebras for the copointed part (Lemma and
  Definition~\ref{lem:copointed}) of~$M_\Lang$.
\end{theorem}
\noindent By the equivalence between the canonical structure and
neighbourhood semantics as outlined in Section~\ref{sec:can-struct},
this result implies Surendonk's strong completeness result for
neighbourhood semantics~\cite{Surendonk97}:
\begin{corollary}[Strong completeness over neighbourhood
  semantics] \label{corollary} Every non-iterative logic
  $\Lang=(\Lambda,\Axioms)$ is (sound and) strongly complete over its
  neighbourhood semantics, i.e.\ over the class of neighbourhood
  frames that satisfy the axioms in~$\Axioms$ as frame conditions.
\end{corollary}
\begin{remark}
  Surendonk's proof~\cite{Surendonk97} shows the existence of a
  suitable superalgebra~$C$ of the powerset algebra of the canonical
  model, going via the first-order model theory of modal algebras,
  specifically via compactness of first-order logic, and demonstrates
  that a suitable neighbourhood structure on the canonical model can
  be inherited from~$C$. Contrastingly, our proof works directly on
  the canonical model, and relies mostly on basic facts on solutions
  of equations in Boolean algebras that are developed from
  Lemma~\ref{boolEq}.
\end{remark}
\section{Application to Deontic Logic}\label{sec:deontics}

Deontic logic is concerned with modalities of obligation, such as
$O\phi$ `$\phi$ is obligatory' and $O(\phi|\psi)$
`given~$\psi$,~$\phi$ is obligatory' (conditional obligation). It is
faced with with specific challenges; e.g., conditional obligations are
defeasible, and it is therefore nontrivial to come with principles of
\emph{factual detachment}, i.e.\ of deriving actual from conditional
obligations, and moreover one needs to avoid the \emph{deontic
  explosion} that would be caused by unrestricted normality of the
obligation modality: If one had an axiom
$(Oa\land Ob)\to O(a\land b)$, then a single dilemma
($Oa\land O\neg a$) would cause impossible obligations ($O\bot$),
making everything obligatory if additionally monotonicity is
imposed. Recent developments in deontic logic often are driven mostly
axiomatically, so that the only available semantics is neighbourhood
semantics.

As an example, we treat axioms for factual detachment proposed by
Straßer~\cite{DBLP:journals/japll/Strasser11}. The full logical
framework uses principles of adaptive logic to govern the actual
factual detachment mechanism; here, we concentrate on the underlying
deontic logics called the \emph{base logics} of the framework. The
logic distinguishes specific types of obligation respectively called
\emph{instrumental} and \emph{proper} (we refer
to~\cite{DBLP:journals/japll/Strasser11} for their philosophical
definition), and has modalities $O(- \mid -)$ (binary conditional
obligation), $O^{i}$ (unary instrumental obligation), $O^{p}$ (unary
proper obligation), and $\bullet^{i}O(- \mid -)$ ,
$\bullet^{p}O(- \mid -)$; the latter two binary modalities serve to
block factual detachment of instrumental and proper obligations from
conditional obligations, respectively. Corresponding dual permission
modalities are denoted by replacing~$O$ with~$P$. Various
axiomatizations are developed as extensions of Goble's logic
\textbf{CPDM}, which is aimed at avoiding the deontic explosion and is
axiomatized in rank~1~\cite{Goble04}. In the online
appendix~\cite{Strasser-Appendix}
to~\cite{DBLP:journals/japll/Strasser11}, it is shown that two such
logics \textbf{CDPM.2d$^+$} and \textbf{CDPM.2e$^+$} are weakly
complete w.r.t.\ neighbourhood semantics when nesting of modalities is
excluded. These logics are non-iterative; they include congruence
rules and various rank-1 axioms that we refrain from listing in full,
and properly non-iterative axioms
\begin{align*}
		&(O(a \mid b) \wedge b \wedge \neg \bullet^{p}O(a \mid b)) \rightarrow O^p a\tag{FDp}\\
		&(O(a \mid b) \wedge b \wedge \neg \bullet^{i}O(a \mid b)) \rightarrow O^i a\tag{FDi}\\
		&(O(a \mid b) \wedge \neg a \wedge b) \rightarrow \bullet^iO(a\mid b)\tag{fV}\\
		& ((P(\neg a \mid b \wedge c) \vee O(\neg a \mid b \wedge c))\tag{Ep}\\
		&\qquad\wedge b \wedge c \wedge P(b \wedge c \mid b) \wedge O(a \mid b)) \rightarrow \bullet^pO(a \mid b)\\
		& ((P(\neg a \mid b \wedge c) \vee O(\neg a \mid b \wedge c)) \tag{oV-Ei}\\
		&\qquad\wedge b \wedge c \wedge O(a \mid b))\rightarrow \bullet^iO(a \mid b)
\end{align*}
where we have converted (Ep) and (oV-Ei) from rules to axioms
(Remark~\ref{rem:rules}). E.g.\ (FDp) says that we can detach a proper
obligation $O^pa$ from a conditional $O(a\mid b)$ if this is not
blocked and $b$ is actually the case, and $(fV)$ say that detaching an
instrumental obligation $O^ia$ from a conditional obligation
$O(a\mid b)$ is blocked if the obligation is factually violated
($\neg a\land b$). By Theorem~\ref{StrongCompletenes}, the fully modal
versions (with nested modalities) of both \textbf{CDPM.2d$^+$} and
\textbf{CDPM.2e$^+$} are strongly complete w.r.t.\ their canonical
coalgebraic semantics (and, by Corollary~\ref{corollary} or cited
previous results~\cite{Surendonk97}, w.r.t.\ their neighbourhood
semantics).

\section{Conclusion and Future Work}

We have shown that every non-iterative modal logic is strongly
complete over a canonical coalgebraic semantics, thus in particular
providing a coalgebraic perspective on the known result that
non-iterative modal logics are strongly complete over neighbourhood
semantics~\cite{Surendonk97}. A fine point in the coalgebraic
semantics is that conceptually, the proof needs to use weakly
copointed functors, equipped with a natural transformation into the
ultrafilter functor instead of the identity functor like copointed
functors, to incorporate non-iterative frame conditions, instead of
copointed functors as one would expect. That is, the natural
generalization of the construction for the rank-1
case~\cite{DBLP:journals/logcom/SchroderP10}, which uses maximally
consistent sets in the so-called 0-1-step logic, produces only a
weakly copointed functor. Ex post, however, our main result then does
imply completeness w.r.t.\ a copointed subfunctor.  We have
illustrated these results on deontic logics allowing factual
detachment~\cite{DBLP:journals/japll/Strasser11}, obtaining that these
logics are strongly complete over their canonical coalgebraic
semantics. It will be interesting to connect our results to
coalgebraic ultrafilter extensions~\cite{DBLP:conf/calco/KupkeKP05}
and the coalgebraic Goldblatt-Thomason
theorem~\cite{DBLP:conf/calco/KurzR07}.

\bibliographystyle{aiml20}
\bibliography{aiml20}

\iffull

\newpage
\Appendix

\section{Omitted Proofs}\label{app:proofs}

\begin{proof*}{Proof of Lemma~\ref{lem:copointed}.}
\noindent We have to show that $Tf(t)\in T_cY$ for $f\colon X\to Y$
and $t\in T_cX$. By naturality of~$\varepsilon$, this amounts to
showing that $\ultra f$ preserves principal ultrafilters. But this is just
naturality of the unit~$\eta$ of the ultrafilter monad. The remaining
claims are then clear.
\end{proof*}

\begin{proof*}{Proof of Lemma~\ref{soundnessImplication}.} Let
  $\sigma$ be an $\form(\Lambda)$-substitution. Let
  $C = (X, \xi)$ be a proper $T$-coalgebras and
  let $\hat{\sigma}$ be the $\pow X$-valuation with
  $\hat{\sigma}(a) = \llbracket \sigma(a) \rrbracket_C$ for all
  $a \in V$. By definition of 0-1-step soundness we have
  $TX \models ^{\zeroone}_X \psi\hat{\sigma}$. We have to show that
  $\psi\sigma$ is valid.  We prove the stronger claim that for  we have
  \begin{equation*}
    \{x \in X \mid \xi(x) \in \llbracket \phi\hat{\sigma}
    \rrbracket^{\zeroone}_X\} = \llbracket \phi\sigma \rrbracket_C
  \end{equation*}
  by induction over $\phi\in\Prop(\Lambda(\Prop(V))\cup V)$. The
  Boolean cases are trivial. The case for modal operators is just by
  expanding definitions: We have
  $\llbracket (L\phi)\hat{\sigma} \rrbracket^{\zeroone}_X= \llbracket
  L \rrbracket (\llbracket \phi\hat{\sigma} \rrbracket)=\llbracket L
  \rrbracket (\llbracket \phi \sigma \rrbracket_C)$, where the last
  step is by induction, and
  $\xi(x) \in \llbracket L \rrbracket (\llbracket \phi \sigma
  \rrbracket_C)$ iff $x\models_C (L\phi)\sigma$.

  For the case $\phi = a \in V$, we have to show
  $\{ x \in X \mid \xi(x) \in \iota(\hat{\sigma}(a))\} =
  \llbracket\sigma(a) \rrbracket_C$. So let $x \in X$. Then
  \begin{align*}
    & \xi(x) \in \iota(\hat{\sigma}(a)) \\
    &\iff \hat{\sigma}(a)\in\varepsilon(\xi(x)) &&\by{definition}\\
    &\iff\hat{\sigma}(a)\in \dot x &&\by{$\xi$ proper}\\
    &\iff x \in \hat{\sigma}(a)\\
    &\iff x \in \llbracket \sigma(a) \rrbracket_C
  \end{align*}
\end{proof*}


\begin{proof*}{Proof of Lemma~\ref{lem:axiom-completeness}}
  Let $\psi$ be a 0-1-step formula over~$\pow X$ such that
  $T_{\Axioms'}X\models^{\zeroone}_X\psi$ One shows analogously to
  \cite[Proposition 23]{DBLP:journals/jlp/Schroder07} that the
  0-1-step logic has the finite (in fact, exponential) model property;
  we can thus assume that~$X$ is finite. Since~$\Axioms'$ mentions
  only finitely many modality, this implies that there are, up to
  propositional equivalence, only finitely many different
  $\Prop(\pow X)$-instances of the axioms in~$\Axioms'$; let $\phi$
  denote the conjunction of these finitely many instances. Then
  $TX\models^{\zeroone}_X\phi\to\psi$. By 0-1-step completeness of
  $\Axioms$, it follows that $\nondash\phi\to\psi$, and hence
  $\nondash[\Lang']\psi$ for $\Lang'=(\Lambda,\Axioms\cup\Axioms')$,
  as required.
\end{proof*}

\begin{proof*}{Proof of Lemma~\ref{weakSubstLemma}}
  Let $\kappa$ be the $U$-substitution such that $\tau = \sigma\kappa$
  and assume $\Phi\sigma \vdash_{PL} \psi\sigma$. Then by the
  substitution lemma of propositional logic it follows that
  $\Phi\sigma\kappa \vdash_{PL} \psi\sigma\kappa$
\end{proof*}

\begin{proof*}{Proof of Lemma~\ref{crr}.} For each equivalence class
  $[a]_\sigma$ of the equivalence relation $\sim_\sigma$ on~$V$ given
  by $a\sim_\sigma b$ iff $\sigma(a) = \sigma(b)$, fix a
  representative $v([a]_\sigma)$, and let $\tau'$ be the
  $W$-substitution defined by $\tau'(a) = \tau(v([a]_\sigma))$. Then
  $\Phi\sigma\vdash_{PL} \psi\sigma$ implies
  $\Phi\tau' \vdash_{PL} \psi\tau'$ by Lemma
  \ref{weakSubstLemma}. Lastly, $\Phi\tau \cup \Psi$ entails
  $\Phi\tau'$ and $\{\psi\tau'\} \cup \Psi $ entails $\psi\tau$.
\end{proof*}

\begin{proof*}{Proof of Lemma~\ref{01stepTruthLemma}.}
  Induction over $\psi$,
  where the cases for Boolean operators are by the Hintikka property of
  maximally consistent sets. The cases for modal operators and formulae of
  the form $\psi \in \pow X$ are by construction.
\end{proof*}

\begin{proof*}{Proof of Lemma~\ref{lem:can-zeroone-complete}}
  We have to show that every one-step consistent formula is
  satisfiable. This is immediate from the 0-1-step Lindenbaum
  lemma~\ref{lem:zeroone-lindenbaum} and the 0-1-step truth
  lemma~\ref{01stepTruthLemma}.
\end{proof*}

\begin{proof*}{Proof of Lemma~\ref{consistency1}.}
  Take $V = \{a_\phi \mid \phi \in \form(\Lambda)\}$, let $\sigma$ be
  the $\form(\Lambda)$-substitution given by $\sigma(a_\phi) = \phi $,
  and let $\hat{\sigma}$ be the $\pow C_\Lang$-valuation given by
  $\hat{\sigma}(a_\phi) = \hat{\phi}$. Lastly, put
  \begin{equation*}
    \Psi =  \{ La_\phi \mid L\phi \in \Phi \} \cup  \{ \neg
    La_\phi \mid \neg L\phi
    \in \Phi \}  \cup \{ a_\phi \mid \phi \in \Phi \}.
  \end{equation*}
  \noindent The claim states that $\Psi\hat{\sigma}$ is is 0-1-step
  consistent. By Lemma~\ref{ci}, this follows from the fact that
  $\Psi\sigma = \Phi$ is consistent.
\end{proof*}
      
\begin{proof*}{Proof of Lemma~\ref{consistency2}.}
  Let
  $V = \{a_\phi \mid \phi \in \form(\Lambda)\} \cup \{a_A \mid A \in
  \pow C_\Lang\}$, and let~$\tau$ be the $\pow C_\Lang$-valuation
  given by $\tau(a_\phi) = \hat{\phi}$ for $\phi \in \form(\Lambda)$
  and $\tau(a_A) = A$ for $A \in \pow C_\Lang$. Put
  \begin{equation*}
    \Psi = \{ La_\phi \mid L\phi \in \Phi \} \cup 
    \{ \neg L a_\phi \mid \neg L\phi \in \Phi \} \cup \{ a_\phi \mid
    \phi \in \Phi \}
  \end{equation*}
  By Lemma \ref{consistency1}, $\Psi\tau$ is 0-1-step consistent. We
  have to show that $\Psi\tau\cup\dot\Phi$ is 0-1-step
  consistent. Assume the contrary, i.e.\
  $\Psi\tau\cup\dot\Phi\nondash\bot$.  Then we have a finite subset
  $\Psi_0\subseteq\Psi$, a finite set
  $\Gamma\subseteq\{a_A\mid A\in\dot\Phi\}$, a finite set~$\Theta_1$
  of axioms, and a finite set $\Theta_2\subseteq\Prop(V)$, which we
  can assume to be instantiated by $\kappa\tau$ for a
  $V$-substitution~$\kappa$ (every subset of~$C_\Lang$ has a name
  in~$V$, and we can disjointly rename variables in axioms to ensure
  that we can use the same substitution~$\kappa$ throughout), such
  that the formulae in $\Theta_2\kappa\tau$ are propositionally valid
  over~$\pow X$, and
  \begin{equation*}
    (\Psi_0 \cup \Gamma \cup (\Theta_1\cup\Theta_2)\kappa)\tau \vdash_{PL} \bot.
  \end{equation*}
  We proceed by induction over $\vert \Gamma \vert$. For
  $\vert \Gamma \vert = 0$ we obtain $\Psi\tau \nondash \bot$,
  contradicting Lemma \ref{consistency1}. Now let
  $\vert \Gamma \vert = n>0$. By Lemma \ref{substLemma}, we have
  $(\Psi \cup \Gamma \cup (\Theta_1\cup\Theta_2)\kappa)\tau' \vdash_{PL} \bot$ for
  any $\Prop(\pow C_\Lang)$-substitution $\tau'$ that solves the
  following system of Boolean equations:
  \begin{itemize}
    
  \item For all subformulae $L \rho, L\rho'$ in
    $(\Psi_0 \cup \Gamma \cup (\Theta_1\cup\Theta_2)\kappa)$ such that
    $L \rho\tau=L\rho'\tau$ in $\Lambda(\pow X)$, we must have
    $L \rho\tau'=L\rho'\tau'$ in $\Lambda(\pow X)$. This amounts to an
    equation $\rho = \rho'$.

  \item For all occurrences of subformulae of the form $\rho$ and
    $\rho'$ in
    $(\Psi_0 \cup \Gamma \cup (\Theta_1\cup\Theta_2)\kappa)$ that do
    not lie beneath a modal operator and such that
    $\rho\tau=\rho'\tau$, we must have $\rho\tau'=\rho'\tau'$ in
    $\pow X$. This amounts to an equation $\rho = \rho'$.

  \end{itemize}
  Now pick $a\in \Gamma$. Define the $\pow C_\Lang$-valuation~$\sigma$
  to be the restriction of~$\tau$ to
  $\Gamma\cup\{a_\phi\mid\phi\in\form(\Lambda)\}$. By
  Lemma~\ref{boolEq}.\ref{item:ba-solve}, since $\tau$ solves the
  above system of equations, there is a
  $\Prop(V_0\cup\Gamma \setminus \{ a\})$-substitution $\sigma'$ on
  $\{a\}$, where $V_0$ is the set of variables occurring in $\Psi$,
  such that the above conditions hold for $\tau' = \sigma'\sigma$. By
  Lemma \ref{boolEq}.\ref{item:monot-solve}, since $\Phi \in \tau(a)$
  it follows that $\Phi \in \llbracket a \sigma'\sigma
  \rrbracket$. It then follows that $a\sigma'\sigma$ is
  $\Lang$-derivable from $(\Psi \cup \Gamma\setminus \{a\})\sigma$:
  Assume without loss of generality that $a\sigma'\sigma$ is in CNF;
  then every clause of $a\sigma'\sigma$ has to contain a literal
  whose interpretation contains $\Phi$, and every such literal is
  contained in $(\Psi\cup\Gamma\setminus\{a\})\sigma$. We thus have
  $(\Psi \cup \Gamma \setminus \{a\})\sigma \nondash \bot$,
  contradicting the induction hypothesis.
\end{proof*}

\begin{proof*}{Proof of Lemma~\ref{existenceLemma} (Existence lemma).}
  We define a coherent proper coalgebra structure
  $\xi:C_\Lang\to M_\Lang C_\Lang$ as follows. Let $\Phi\in
  C_\Lang$. By Lemma~\ref{consistency2} and the 0-1-step Lindenbaum
  lemma (Lemma~\ref{lem:zeroone-lindenbaum}),
  there is $\Psi\in M_\Lang C_\Lang$ such that
  \begin{equation*}
    \Psi \supseteq \{ L\hat{\phi} \mid L\phi \in \Phi \}  \cup \{
    \neg L\hat{\phi} \mid \neg L\phi \in \Phi \}  \cup \; \dot{\Phi}
  \end{equation*}
  We put $\xi(\Phi)=\Psi$. It is then clear that~$\xi$ is coherent; it
  remains to show that~$\xi$ is proper, i.e.\ that $A\in\Psi$ iff
  $\Phi\in A$, for $A\subseteq C_\Lang$: `If' holds by
  construction. For `only if', assume $\Phi\notin A\in\Psi$. But then
  $\Phi\in C_\Lang\setminus A$, so $C_\Lang\setminus A\in\Psi$,
  contradicting 0-1-step consistency of~$\Psi$.
\end{proof*}

\begin{proof*}{Proof of Theorem~\ref{StrongCompletenes}.}
  The existence lemma shows that there is a canonical coalgebra
  $C = (C_\Lang, \xi)$ in which, by the truth lemma, every maximally
  consistent set is satisfiable. By the Lindenbaum lemma every
  consistent set is contained in such a maximally consistent set and
  therefore also satisfiable.
\end{proof*}

\begin{proof*}{Details for the proof of Corollary~\ref{corollary}.}
  The weak copoint on $\neighb_\Lambda\times\ultra$ is just the second
  projection, while predicate liftings are defined like in
  Example~\ref{expl:logics}.\ref{nft}, on the first projection. The
  precise definition of $\neighb_\Lang$ is
  \begin{align*}
    \neighb_\Lang X = \{(N,\alpha)\in \neighb_\Lang X\times\ultra X\mid
                      (N,\alpha)\models^{\zeroone}_X\psi\sigma&\text{ for all $\psi\in\Axioms$}\\
    &\text{and all $\pow X$-substitutions~$\sigma$}\}.
  \end{align*}
  To see that $\neighb_\Lang$ is a subfunctor of
  $\neighb_\Lambda\times\ultra$, let $f\colon X\to Y$, and let
  $(N,\alpha)\in\neighb_\Lang X$; we have to show
  $(\neighb_\Lambda f(N),\ultra f(\alpha))\in\neighb_\Lang Y$. This
  follows from the fact that
  \begin{equation*}
    (\neighb_\Lambda f(N),\ultra f(\alpha))\models^{\zeroone}_Y\phi \text{ iff } (N,\alpha)\models^{\zeroone}_X\phi\sigma_f,
  \end{equation*}
  for all $\phi\in\Prop(\Lambda(\pow Y)\cup\pow Y)$, where $\sigma_f$
  is the $\pow X$-substitution given by $\sigma_f(B)=f^{-1}[B]$ for
  $B\subseteq Y$. This is shown by induction on~$\phi$, with trivial
  Boolean cases and using naturality of predicate liftings in the
  modal cases; the base case $B\subseteq Y$ is just by definition of
  $\ultra f$. 
  It is clear that $\Lang$ is sound for the $\Lambda$-structure based
  on~$N_\Lang$ obtained by restricting the original predicate
  liftings.

  Naturality of the transformation~$\theta$ is clear. The inverse
  transformation $\theta^{-1}$ is defined on
  $(N,\alpha)\in\neighb_\Lang X$ by
  \begin{equation*}
    \theta^{-1}(N,\alpha) = \{\psi\in\Prop(\Lambda(\pow X)\cup\pow X)\mid (N,\alpha)\models\psi\}.
  \end{equation*}
  Note that $\theta^{-1}(N,\alpha)$ is satisfied by $(N,\alpha)$,
  hence 0-1-step consistent by soundness; it is then clear that
  $\theta^{-1}(N,\alpha)$ is maximally consistent, i.e.\ in
  $M_\Lang X$. One easily checks that $\theta$ and $\theta^{-1}$ are
  really mutually inverse.

  It is clear that~$\theta$ commutes with the respective predicate
  liftings for $N_\Lang$ and $M_\Lang$; it follows that every proper
  $M_\Lang$-coalgebra $(X,\xi)$ satisfies the same $\Lambda$-formulae
  as the induced proper $N_\Lang$-coalgebra $(X,\theta\circ\xi)$, and
  similarly in the converse direction; thus, $\Lang$ is sound and
  strongly complete over proper $N_\Lang$-coalgebras. Since properness
  of an $N_\Lang$-coalgebra $(X,\xi)$ implies that the second
  component $\alpha$ of $\xi(x)=(N,\alpha)$ is uniquely determined (as
  $\dot x$) for every $x\in X$, proper $N_\Lang$-coalgebras are just
  $\Lambda$-neighbourhood frames satisfying~$\Axioms$.
\end{proof*}
\fi

\end{document}